\documentclass[12pt]{amsart}
\usepackage[sans]{dsfont}
\usepackage{amsfonts,amssymb,amsbsy,amsmath,amsthm,enumerate}

\numberwithin{equation}{section}

\newtheorem{theorem}{Theorem}[section]

\newtheorem{lemma}[theorem]{Lemma}

\newtheorem{follow}[theorem]{Corollary}

\newtheorem{pr}[theorem]{Proposition}
\theoremstyle{definition}

\newcommand{\bel}{\begin{equation} \label}
\newcommand{\ee}{\end{equation}}
\newcommand{\one}{\mathds{1}}

\newcommand{\hpe}{H_\perp}
\newcommand{\bhpe}{{\mathcal H}_\perp}
\newcommand{\hpa}{H_\parallel}
\newcommand{\bhpa}{{\mathcal H}_\parallel}

\newcommand{\E}{{\mathcal E}}

\newcommand{\rd}{{\mathbb R}^{2}}
\newcommand{\re}{{\mathbb R}}
\newcommand{\R}{{\mathbb R}}
\newcommand{\C}{{\mathbb C}}
\newcommand{\Z}{{\mathbb Z}}
\newcommand{\N}{{\mathbb N}}

\def\beq{\begin{equation}}
\def\eeq{\end{equation}}
\newcommand{\bea}{\begin{eqnarray}}
\newcommand{\eea}{\end{eqnarray}}
\newcommand{\beas}{\begin{eqnarray*}}
\newcommand{\eeas}{\end{eqnarray*}}

\DeclareMathOperator*{\esssup}{ess\,sup}

\begin{document}
\title[Surface Lifshits tails for random quantum Hamiltonians]{Surface Lifshits tails for random quantum Hamiltonians}

\author[W.~Kirsch]{Werner Kirsch}
\author[G.~Raikov]{Georgi Raikov}

\begin{abstract}\setlength{\parindent}{0mm}
We consider Schr\"{o}dinger operators on $L^{2}(\R^{d})\otimes L^{2}(\R^{\ell})$ of the form
\begin{align*}
  H_{\omega}~=~H_{\perp}\otimes I_{\parallel} + I_{\perp} \otimes \hpa + V_{\omega},
\end{align*}
where $H_{\perp}$ and $H_{\parallel}$ are Schr\"{o}dinger operators on $L^{2}(\R^{d})$ and $L^{2}(\R^{\ell})$ respectively, and
\begin{align*}
    V_\omega(x,y) : = \sum_{\xi \in \Z^{d}} \lambda_\xi(\omega) v(x - \xi, y), \quad x \in \re^d, \quad y \in \re^\ell,
\end{align*}
is a random `surface potential'.
We investigate the behavior of the integrated density of surface states of $H_{\omega}$ near the bottom of the spectrum and near internal band edges.

The main result of the current paper is that, under suitable assumptions, the behavior of the integrated density of surface states of $H_{\omega}$ can be
read off from the integrated density of states of a reduced Hamiltonian $H_{\perp}+W_{\omega}$ where $W_{\omega}$ is a quantum mechanical average of $V_{\omega}$ with respect to $y \in \re^\ell$.
We are particularly interested in cases when
$H_{\perp}$ is a magnetic Schr\"{o}dinger operator, but we also recover some of the results from \cite{kw} for non-magnetic $H_{\perp}$.
\end{abstract}

\maketitle



\section{Introduction}
\label{s0} \setcounter{equation}{0}

The integrated density of states  is an important quantity in solid states physics. For periodic and ergodic Schr\"{o}dinger operators the integrated density of states has been the object of
intense investigation in the mathematical physics literature for over more than thirty years by now (see e.\,g. the book \cite{pf} or the survey \cite{kme}). In particular,
the behavior near the bottom of the spectrum and near internal band edges has been investigated.

Define the random Schr\"{o}dinger operator $H_{\omega}$ on $L^{2}(\R^{d})$ by
\begin{align}
   H_{\omega}~:=~H(A)+V_{\omega}
\end{align}
where $H(A)$ is the Laplacian with a magnetic potential $A$ and $V_{\omega}$ is a (scalar) random potential. Important examples for $V_{\omega}$ are \emph{Poissonian random potentials}
(see e.\,g. \cite{pf}) and \emph{alloy-type potentials}. We will deal mainly with the latter type of random potentials in this article. An alloy-type potential is of the form (see e.\,g. \cite{km0})
\begin{align}
   V_\omega(x) : = \sum_{\xi \in \Z^{d}} \lambda_\xi(\omega) v(x - \xi),
\end{align}
with independent, identically distributed random variables $\lambda_{\xi}$ on a probability space $\left(\Omega,\mathcal{A},\mathbb{P}\right)$. The function $v$ is called the \emph{single site potential}.

The \emph{integrated density of states} $N(E)$ (see e.\,g. \cite{km}) can be defined for such operators $H_{\omega}$ by
\begin{align}\label{dos}
   N(E)~=~\lim_{L\to\infty} \frac{1}{L^{d}}\;\;{\rm Tr}\,\left( \one_{(-\infty,E)}(H_{\omega, C_{L}^{d}}^{D})\right)
\end{align}
where $C_{L}^{d}$ is the cube around the origin in $\R^{d}$ of side length $L$, $H_{\omega, C_{L}^{d}}^{D}$ is the operator $H_{\omega}$ restricted to $C_{L}^{d}$ with Dirichlet boundary conditions  and $\one_{(-\infty,E)}(H)$ denotes the spectral projection for the operator $H$. It is known that the spectrum of $H_{\omega}$ coincides with the set of growth points of $N$.

For vanishing magnetic potential $A=0$,
the integrated density of states $N(E)$, as a rule, decays exponentially fast near the bottom $E_{1}$ of the spectrum, in fact on a double logarithmic scale
\begin{align}\label{lifbeh}
   N(E)~\sim~e^{-C\,(E-E_{1})^{-\gamma}}\qquad\text{as}\quad E\downarrow E_{1}\,.
\end{align}

The exponent $\gamma$ is called the \emph{Lifshits exponent} and the behavior \eqref{lifbeh} is called the \emph{Lifshits behavior}. The Lifshits exponent
is known to be $\gamma=\frac{d}{2}$ if the single site potential $v$ decays faster than $| x |^{-(d+2)}$ near infinity. If $v$ decays like $| x |^{-\kappa}$ for
$d<\kappa\le d+2$ then $\gamma=\frac{d}{\kappa - d}$.

The same behavior is known at `non-degenerate' internal band edges (see \cite{klopp1}) while
for `degenerate' internal band edges other Lifshits exponents may occur (\cite{kloppW}).\\

The presence of a constant magnetic field changes the behavior of $N(E)$ drastically, already for the free operators $H(A)$. Suppose that the dimension $d$ equals $2$ and the magnetic field $B={\rm curl}\, A$ is
constant. For this case, the integrated density of states of $H(A)$ has a jump at the bottom of the spectrum as long as $B\not=0$ while for $B=0$ the integrated density of states behave like
$E^{\,d/2}$ as $E\downarrow \inf \sigma(H(0))$.

In \cite{klorai} and \cite{klopp} it was shown that for constant magnetic field $B\not=0$ in two dimensions the Lifshits exponent is  $\gamma=\frac{2}{\kappa-2}$ if $v(x)$ behaves like
$|x|^{-\kappa}$ near infinity for \emph{all} $\kappa>2$. If $v(x)$ has at least Gaussian decay then the integrated density of states behaves like $E^{\,|\ln E|}$ on a double logarithmic scale.
We note that analogous results were obtained earlier for Poissonian random potential in \cite{bhkl,erd,erd1}.\\

In the current paper we consider Schr\"{o}dinger operators on $L^{2}(\R^{d})\otimes L^{2}(\R^{\ell})$ of the form
$$
   H_{0}=H_{\perp}\otimes I_{_\parallel} + I_{\perp} \otimes \hpa
   $$
   and
   $$
  H_{\omega}=H_{0} + V_{\omega},
$$
where $H_{\perp}$ and $H_{\parallel}$ are Schr\"{o}dinger operators on $L^{2}(\R^{d})$ and $L^{2}(\R^{\ell})$ respectively, and
\begin{align*}
    V_\omega(x,y) : = \sum_{\xi \in \Z^{d}} \lambda_\xi(\omega) v(x - \xi, y), \quad x \in \re^d, \quad y \in \re^\ell,
\end{align*}
is a random `surface' potential of alloy-type.\\

Suppose for this introduction that $\hpe$ has purely essential spectrum with $\inf\sigma(\hpe)=0$ and $\hpa$ is lower bounded  and has eigenvalues
$E_{j}$ below $\E:=\inf\sigma_{{\rm ess}}(\hpa)$. Assume furthermore that both $v$ and $\lambda_{\xi}$ are non-negative and $\mathbb{P}(\lambda_{\xi}<\varepsilon)>0$ for all $\varepsilon>0$.

The operator is not ergodic with respect to $\mathbb{Z}^{d+\ell}$ but merely with respect to $\mathbb{Z}^{d}$. Nevertheless,
one can prove that the spectrum of $H_{\omega}$ is non-random and the discrete spectrum is empty almost surely (see \cite{ekss}). The integrated density of states $N(E)$ for $H_{\omega}$
can be defined by
\begin{align}
   N(E)~=~\lim_{L\to\infty} \frac{1}{L^{d+\ell}}\;\;{\rm Tr}\,\left( \one_{(-\infty,E)}(H_{\omega, C_{L}^{d+\ell}}^{D})\right)\,,
\end{align}
which is just equation \eqref{dos} with the dimension adjusted.

Since the operator $H_{\omega}$ is not ergodic with respect to $\mathbb{Z}^{d+\ell}$, we can not conclude that the spectrum coincides with the set of growth points of $N$.
In fact, $N(E)=0$ for $E<\E$, but for any $\eta\in\sigma(\hpe)$ and any $j$ we have $\eta+E_{j}\in\sigma(H_{\omega})$ almost surely.

Intuitively speaking, this means that the spectrum around such point is `not dense enough', in the sense that the number $N_{L}(E)$ of eigenvalues of $H_{\omega, C_{L}^{d+\ell}}^{D}$
below $E<\E$ does \emph{not} grow as fast as the volume of $C_{L}^{d+\ell}$.
It is quite reasonable to expect that $N_{L}(E)$ grows rather like $L^{d}$ in the energy region below $\E$.

Thus, we define
\begin{align}
   \nu_{V}(E)~=~\lim_{L\to\infty} \frac{1}{L^{d}}\;\;{\rm Tr}\,\left( \one_{(-\infty,E)}(H_{\omega, C_{L}^{d}\times\R^{\ell}}^{D})\right)
\end{align}
for $E<\E$. In fact, it turns out, that $\nu_{V}(E)$ is well defined under reasonable assumption on $H_{\omega}$. This quantity is called
the \emph{integrated density of surface states}.
The integrated density of surface states was already considered in \cite{ekss} and \cite{ekss1}. In this paper we define $\nu_{V}(E)$ only for $E<\E$. For a discussion of
$\nu_{V}(E)$ for arbitrary $E$ see \cite{ekss,ekss1}.
In the paper \cite{kw} Lifshits tails for the integrated density of surface states were investigated for Schr\"{o}dinger operators without magnetic fields and at the bottom of the spectrum.
We are particularly interested in cases when
$H_{\perp}$ is a magnetic Schr\"{o}dinger operator, but we also recover some known results from \cite{kw} for non-magnetic $H_{\perp}$.
We investigate the behavior of the integrated density of surface states of $H_{\omega}$ near the bottom of the spectrum and near internal band edges.

The main result of the current paper is that under suitable assumptions the behavior of the density of surface states of $H_{\omega}$ can be
read off from the density of states of a reduced Hamiltonian $H_{\perp}+W_{\omega}$ where $W_{\omega}$ is a quantum mechanical average of $V_{\omega}$ with respect to $y \in \R^{\ell}$.
More precisely, if $\psi_{1}$ denotes the ground state of $H_{\parallel}$, then
\begin{align*}
   W_{\omega}(x)~&=~\langle V_{\omega}(x,\cdot)\,\psi_{1},\psi_{1} \rangle\\
   &=~\int_{\R^{\ell}}\;V_{\omega}(x,y)\,|\psi_{1}(y)|^{2}\,dy.
\end{align*}

In particular, we prove that $H_{\omega}$ admits Lifshits tails if $H_{\perp}+W_{\omega}$ does.\\

The article is organized as follows. In the next section we give formal definitions of the operators we deal with, and discuss some of the particular examples we consider important. In Section 3 we prove the existence of the integrated density of surface states, and in Section 4 we estimate it in terms of the integrated density of (bulk) states for a reduced random ergodic operator. Finally, in Section 5 we apply the estimates obtained in Section 4 in order to study the Lifshits tails of the integrated density of surface states for particular random quantum Hamiltonians.

\section{Setting of the problem}
\label{s1} \setcounter{equation}{0}
Let $d \in \N$ and $B = \left\{B_{jk}\right\}_{j,k=1}^d$ be an antisymmetric real matrix. Define the vector field $A = (A_1,\ldots,A_d) : \R^d \to \R^d$ by
$$
A_j(x) : = -\frac{1}{2} \sum_{k=1}^d B_{jk} x_k, \quad j=1,\ldots,d, \quad  x = (x_1,\ldots,x_d) \in \R^d.
$$
Then $B_{jk} = \frac{\partial A_k}{\partial x_j} - \frac{\partial A_j}{\partial x_k}$, $j,k =1,\ldots,d$. Thus, in the sequel $B$  will play the role of a constant magnetic field, while $A$ is a magnetic potential generating $B$. Set
$$
2m : = {\rm dim}\,{\rm Ran}\,B, \quad n : = {\rm dim}\,{\rm Ker}\,B,
$$
so that $d = 2m + n$. Note that we do not exclude the possibility that $m=0$, i.e. $B = 0$; in particular, this is the case if $d=1$. \\
Assume $m > 0$. Let the numbers $b_1 \geq \ldots \geq b_m > 0$ be such that the non-zero eigenvalues of $B$, counted with their multiplicities, coincide with $\pm i b_j$, $j=1,\ldots,m$. Set $\beta : = \sum_{j=1}^m b_j$. If $m = 0$, then $\beta : = 0$. Thus, for all $m \geq 0$, we have $\beta  = {\rm Tr}\,(iB)_+$.\\
 Define the operator $\hpe = \hpe(B): = (i\nabla + A)^2 - \beta$ as the self-adjoint operator generated in the Hilbert space $\bhpe : = L^2(\R^d)$ by the closure of the quadratic form
$$
\int_{\R^d} \left(|i\nabla u  + Au|^2 - \beta |u|^2\right) dx, \quad u \in C_0^\infty(\R^d).
$$
Thus $\hpe$ is just the (shifted) $d$-dimensional Schr\"odinger operator with constant (possibly vanishing) magnetic field.
It is well known that $\hpe$ is essentially self-adjoint on $C_0^\infty(\R^d)$  (see \cite{ls}). Note that the operators
$\hpe(B)$ and $\hpe(-B)$ are anti-unitarily equivalent under complex conjugation, so that their spectra coincide. \\
Let us describe the spectrum $\sigma(\hpe)$ of $\hpe$.
Introduce the (shifted) Landau levels
$$
\Lambda_0 = 0,
$$
$$
\Lambda_{q+1} = \inf{\left\{2\sum_{j=1}^m b_j \ell_j, \; \ell_j \in \Z_+, \; j=1,\ldots,m \; \Big| \; 2\sum_{j=1}^m b_j \ell_j > \Lambda_q\right\}}, \quad q \in \Z_+.
$$
If $n=0$, i.e. if the magnetic field $B$ has a full rank, then $\sigma(\hpe) = \cup_{q=0}^\infty\left\{\Lambda_q\right\}$
and each Landau level $\Lambda_q$, $q \in \Z_+$, is an eigenvalue of $\hpe$ of infinite multiplicity. If $n \geq 1$,
then $\sigma(\hpe)$ is purely absolutely continuous, and
$\sigma(\hpe) = [0,\infty)$. Note however, that if $m \geq 1$, i.e. $B \neq 0$, then the higher Landau levels $\Lambda_q$, $q \in \N$, play the role
of thresholds within $\sigma(\hpe)$, while in the case $m=0$ the only threshold is the origin. \\

Next, let $\bhpa$ be a separable Hilbert space with scalar product $\langle \cdot , \cdot\rangle_{\bhpa}$ and norm $\|\cdot\|_{\bhpa}$, and let $\hpa$ be a linear operator, self-adjoint in $\bhpa$. Assume that
    \bel{d1}
    -\infty < \inf{\sigma(\hpa)} < {\mathcal E} : = \inf {\sigma_{\rm ess}(\hpa)} \leq \infty.
    \ee
    The first inequality in \eqref{d1} just means that $\hpa$ is lower bounded, while the second one implies that there is a number $r \in \left\{1,\ldots,\infty\right\}$ of discrete eigenvalues of $\hpa$ below the bottom ${\mathcal E}$ of its essential spectrum. For notational convenience set
    $$
    {\mathcal J} : = \left\{
    \begin{array} {l}
    \{1\ldots,r\} \quad {\rm if} \quad r<\infty, \\
    \N \quad {\rm if} \quad r = \infty.
    \end{array}
    \right.
    $$
    Let $\left\{E_j\right\}_{j \in {\mathcal J}}$ be the non-decreasing sequence of the eigenvalues of $\hpa$ lying in $(-\infty, {\mathcal E})$.
    If $r=\infty$, then $\lim_{j \to \infty} E_j = {\mathcal E}$.
     If $r < \infty$, we  occasionally  set $E_{r+1} = {\mathcal E}$. Let $\left\{\psi_j \right\}_{j \in {\mathcal J}}$ be an associated orthonormal system of eigenfunctions satisfying
    $$
    \hpa \psi_j = E_j \psi_j, \quad \langle \psi_j, \psi_k \rangle_{\bhpa} = \delta_{jk}, \quad j,k \in {\mathcal J}.
    $$
 Denote by $I_\perp$ (resp., by  $I_\parallel$) the identity in $\bhpe$ (resp., in $\bhpa$). Define the operator
 $$
 H_0 : = \hpe \otimes I_\parallel + I_\perp \otimes \hpa
 $$
 as the closure of the operator defined on ${\rm Dom}(\hpe) \otimes  {\rm Dom}(\hpa)$. Thus, $H_0$ is self-adjoint in the Hilbert space
 ${\mathcal H} : = \bhpe \otimes \bhpa$ (see e.g. \cite[Theorem VIII.33 a]{rs1}). It is well known that the space ${\mathcal H}$ is isometrically isomorphic to $L^2(\re^d; \bhpa)
 = \int_{\re^d}^\oplus \bhpa dx$ under the mapping ${\mathcal K}$, defined originally by ${\mathcal K} : g(x)\otimes\psi \mapsto  g(x)\psi$, $x \in \re^d$,
 for $g \in {\bhpe} = L^2(\re^d)$ and $\psi \in {\bhpa}$, extended then by linearity to finite sums $\sum_j g_j \otimes \psi_j$ with $g_j \in \bhpe$, $\psi_j \in \hpa$,
 and finally extended by continuity to a unitary operator from ${\mathcal H}$ to $L^2(\re^d; \bhpa)$. In the sequel, we will systematically identify
 ${\mathcal H}$ with $L^2(\re^d; \bhpa)$, omitting ${\mathcal K}$ and ${\mathcal K}^*$ in the notations.\\

 If $n \geq 1$, then $\sigma(H_0) = [E_1, \infty)$ is purely absolutely continuous (see e.g. \cite[Subsection 8.2.3]{abmg}), while if $n=0$, then
    \bel{d28}
 \sigma(H_0) \cap (-\infty, {\mathcal E}) = \bigcup_{j \in {\mathcal J}, \, q \in \Z_+ \, : \, E_j + \Lambda_q < {\mathcal E}} \left\{E_j + \Lambda_q\right\},
    \ee
 and the energies $E_j + \Lambda_q < {\mathcal E}$ are isolated eigenvalues of $H_0$ of infinite multiplicity.\\
 Further, we introduce a random perturbation of the operator $H_0$. Let $(\Omega, {\mathcal A}, {\mathbb P})$ be a probability space, ${\mathbb G} = \re$ or ${\mathbb G} = \Z$, and let
 ${\mathbb T} : = \left\{{\mathcal T}_\xi\right\}_{\xi \in {\mathbb G}^d}$ be an ergodic group of measure preserving automorphisms of $\Omega$, homomorphic to ${\mathbb G}^d$.
 Ergodicity of $\mathbb{T}$ means that any set $A\in\mathcal{A}$ which is invariant under all $\mathcal{T}_{\xi}$ has probability $\mathbb{P}(A)=0$ or $\mathbb{P}(A)=1$.\\

    Denote by ${\mathcal L}(\bhpa)$ the space of linear bounded operators in $\bhpa$. Introduce the function
    $$
    \Omega \times \re^d \ni (\omega,x) \mapsto V_\omega(x) \in {\mathcal L}(\bhpa).
    $$
    We suppose that $V_\omega$ satisfies the following assumptions:\\
    $\bf{H_1}$: For each $f,g \in \bhpa$ the function
    $$
    \Omega \times \re^d \ni (\omega,x) \mapsto \langle V_\omega(x) f, g\rangle_{\bhpa} \in \C
    $$
    is measurable with respect to the $\sigma$-algebra ${\mathcal A} \times {\mathcal B}$, where ${\mathcal B}$ is the $\sigma$-algebra of Borel subsets of $\re^d$.\\
    $\bf{H_2}$: We have
    \bel{d2}
    M : = \esssup_{(\omega,x) \in \Omega \times \re^d} \|V_\omega(x)\|_{{\mathcal L}(\bhpa)} < \infty.
    \ee
     $\bf{H_3}$: For almost every $(\omega,x) \in \Omega \times \re^d$, the operator $V_\omega(x)$ is self-adjoint and non-negative in $\bhpa$.\\
     $\bf{H_4}$: The family of operators $V_\omega : = \int_{\re^d}^\oplus V_\omega(x) dx \in {\mathcal L}({\mathcal H})$, $\omega \in \Omega$, is ergodic with respect to the group ${\mathbb T}$, i.e. we have
     $$
     V_\omega(x-\xi) = V_{{\mathcal T}_\xi \omega}(x), \quad x \in \re^d, \quad \xi \in {\mathbb G}^d, \quad {\mathcal T}_\xi \in {\mathbb T}.
     $$
     Introduce the family of operators
     $$
     H_\omega : = H_0 + V_\omega, \quad \omega \in \Omega.
     $$
     By $\bf{H_1}$ --  $\bf{H_3}$, the operator $H_\omega$ is well defined on ${\rm Dom}\,(H_0)$ and self-adjoint in ${\mathcal H}$ for almost every $\omega \in \Omega$. \\

     Let us now describe our leading example. We assume in it that $d=2$, $m=1$, and hence $n=0$. We suppose without loss of generality that $B_{12} > 0$,
     and set $b: = B_{12} = b_1$. Then we have $\Lambda_q = 2bq$, $q \in \Z_+$, (see e.g. \cite{f,l}). Further, we assume that $\bhpa = L^2(\re)$, and
     $\hpa : = - \frac{d^2}{dy^2} + u$, i.e. $\hpa$ is the 1D Schr\"odinger operator with appropriate real-valued potential $u$.
     More precisely, $\hpa$ is the self-adjoint operator generated in $L^2(\re)$ by the closure of the quadratic form
     \bel{d3}
     \int_\re\left(|f'|^2 + u |f|^2\right) dy, \quad f \in  C_0^\infty(\re).
     \ee
     In order that the quadratic form \eqref{d3} be closable and lower bounded in $L^2(\re)$, and that inequalities \eqref{d1} hold true,
     we have to impose additional conditions on $u$. For instance, we may assume that $u \in  L^1(\re) + L^\infty_\epsilon(\re)$, and that there exist
     a constant $c \in (0,\infty)$, and an open non-empty set $S \subset \R$, such that
     $$
     u(y) \leq - c \one_S(y), \quad y \in \re;
     $$
     here and in the sequel $ \one_S$ is the characteristic function of a given set $S$. Another possibility is to assume that $u \in L^1(\re; (1+x^2)dx)$, $u \neq 0$,  and $\int_\re u(y) dy \leq 0$. In both cases, the quadratic form \eqref{d3} is closable and lower bounded, $\sigma_{\rm ess}(\hpa) = [0,\infty)$, and the discrete spectrum $\sigma_{\rm disc}(\hpa)$ of $\hpa$ is non-empty and simple (see e.g. \cite{bir, s1}). A certain generalization of these assumptions is the case where $u = -\alpha \delta$ with  fixed $\alpha > 0$, i.e.  $\hpa$  is the self-adjoint operator generated in $L^2(\re)$ by the closed lower bounded quadratic form
     $$
     \int_\re \,|f'|^2 \, dy - \alpha |f(0)|^2, \quad f \in  {\rm H}^1(\re),
     $$
     where  ${\rm H}^1(\re)$ denotes the first-order Sobolev space on $\re$. In this case again $\sigma_{\rm ess}(\hpa) = [0,\infty)$, and an explicit calculation shows that $\sigma_{\rm disc}(\hpa) = \left\{-\frac{\alpha^2}{4}\right\}$, and $-\frac{\alpha^2}{4}$ is a simple eigenvalue of $\hpa$
     (see e.g. \cite[Chapter I.3, Theorem 3.1.4]{aghkh}). Finally, we might assume that $0 \leq u \in L^\infty_{\rm loc}(\re)$ and  we have $\lim_{|t| \to \infty} \int_{t-\varepsilon}^{t^+\varepsilon}u(y) dy = \infty$ for a given $\varepsilon > 0$. Then again the quadratic form \eqref{d3} is closable and lower bounded, but now the spectrum of $\hpa$ is purely discrete and simple (see e.g. \cite{bir}).
     Thus, in our leading example
     \bel{d10}
H_0  = \left(-i\frac{\partial}{\partial x_1} + \frac{b x_2}{2}\right)^2 + \left(-i\frac{\partial}{\partial x_2} - \frac{b x_1}{2}\right)^2 -
\frac{\partial^2}{\partial y^2} + u(y) - b.
    \ee
    Hence, in this case, $H_0$ is the (shifted) 3D Schr\"odinger operator with constant magnetic field which could be identified with
    the vector ${\bf B} = (0,0,b)$ and electric potential $u = u(y)$; then the electric field ${\bf E} = (0,0,-u'(y))$
    is parallel to the magnetic field ${\bf B}$. Moreover, $(x_1,x_2) \in \re^2$ are the variables on the plane perpendicular to ${\bf B}$,
    while $y \in \re$ is the variable along ${\bf B}$, which explains our notations $\hpe$ and $\hpa$.

    The spectral properties of the operator $H_0$ in \eqref{d10}, perturbed by a
    rapidly decaying non-random electric potential $V$, were discussed in \cite{abbfr}. The problems attacked there were the
    accumulation of resonances and the singularities of the spectral shift function for the pair $(H_0 + V, H_0)$ at the points
    $2bq + E_j$, $q \in \Z_+$, $j \in {\mathcal J}$.\\

    In our other example of $H_0$, which is a special case of the unperturbed operator considered in \cite{kw}, we assume $B = 0$.
    Further, we suppose that $\bhpa = L^2(\re^{\ell})$ with $\ell \in \N$, while $\hpa$ is the self-adjoint operator generated in $\bhpa$
    by the closure of the quadratic form
    \bel{d50a}
    \int_{\re^{\ell}} \left(|\nabla f|^2 + U|f|\right)^2 dy, \quad f \in C_0^\infty(\re^{\ell}),
    \ee
    where $U : \re^{\ell} \to \re$ is an appropriate potential. If, for instance, $U \in L^p(\re^{\ell}) + L_\epsilon^\infty(\re^{\ell})$ with $p=1$ if $\ell =1$, $p>1$ if $\ell =2$,
    and $p=\ell/2$ if $\ell \geq 3$, then
    the quadratic form in \eqref{d50a} is lower bounded and closable, and $\sigma_{\rm ess}(\hpe) = [0,\infty)$. Under suitable assumptions on $U$,
    the discrete spectrum of $\bhpe$ is non-empty, and its smallest eigenvalue $E_1$ is simple (see e.g. \cite{rs4}). Thus, in our second example,
    \bel{d51a}
H_0  = - \Delta_x - \Delta_y + U(y).
    \ee
    {\em Remark}: The operator $H_0$ admits further extensions. For instance, if $Q_{\rm per} \in L^{\infty}(\re^{d}, \re)$ is a
    $\Z^{d}$-periodic function, then we could replace $-\Delta_x$ by
    $-\Delta_x + Q_{\rm per}(x)$.\\

    Next, in both our examples the random perturbation $V_\omega$ of $H_0$   is  the multiplier by an alloy-type electric potential
\bel{2}
    V_\omega(x,y) : = \sum_{\xi \in \Z^{d}} \lambda_\xi(\omega) v(x - \xi, y), \quad \omega \in \Omega, \quad x \in \re^{d}, \quad y \in \re^{\ell},
    \ee
    with $d = 2$ and $\ell = 1$ in the case of a perturbation of \eqref{d10}, and arbitrary $d$, $\ell \in \N$ in the case of a perturbation of \eqref{d51a}.
    The single-site potential $v$ in \eqref{2} is supposed to be Lebesgue measurable and to satisfy
    $$
    c_0^- \one_S(x,y) \leq v(x,y) \leq c_0^+ (1 + |x|)^{-\varkappa}, \quad x \in \re^{d}, \quad  y \in \re^{\ell},
    $$
    with $\varkappa > d$, $0 < c_0^- \leq c_0^+ < \infty$, and an open non-empty set $S \subset \re^{d + \ell}$,
      while the coupling constants $\lambda_\xi$, $\xi \in \Z^{d}$, are i.i.d random variables on $\Omega$ which almost surely are non-negative and bounded. \\

    \section{Existence of the integrated density of surface states}
    \label{s2} \setcounter{equation}{0}
    Our next goal is to introduce the integrated density of surface states for the operator $H_{\omega}$ in the general setting.
    Let ${\mathcal O} \subset \re^d$ be a bounded, open, non-empty set.
    Define $H^{D}_{\perp, {\mathcal O}}$ as the self-adjoint operator generated in $L^2({\mathcal O})$ by the closed non-negative quadratic form
    \bel{d6}
    \int_{{\mathcal O}} \left(|i\nabla f + Af|^2   - \beta)|f|^2\right) dx, \quad f \in {\rm H}_0^1({\mathcal O}),
    \ee
    where ${\rm H}_0^1({\mathcal O})$ is the closure of $C_0^\infty({\mathcal O})$ in ${\rm H}^1({\mathcal O})$. Due to the compactness of the embedding of  ${\rm H}_0^1({\mathcal O})$ into $L^2({\mathcal O})$, the spectrum of the operator $H^{D}_{\perp, {\mathcal O}}$ is purely discrete: Moreover, as already mentioned,
    \bel{d5}
    \inf \sigma(H^{D}_{\perp, {\mathcal O}}) \geq 0.
    \ee

    Denote by $I_{\perp, {\mathcal O}}$ the identity  in $L^2({\mathcal O})$. Define the operator
    $$
    H^{D}_{0, {\mathcal O}} : = H^{D}_{\perp, {\mathcal O}} \otimes I_\parallel + I_{\perp, {\mathcal O}} \otimes \hpa,
    $$
     self-adjoint in $L^2({\mathcal O}) \otimes \bhpa$. Evidently, the spectrum of $H^{D}_{0, {\mathcal O}}$ on $(-\infty, {\mathcal E})$ is discrete.\\
     Further, due to ${\bf H_1}$ - ${\bf H_3}$, the operator $H^{D}_{\omega, {\mathcal O}} : = H^{D}_{0, {\mathcal O}} + V_{\omega}$
     is almost surely well defined on ${\rm Dom}(H^{D}_{0, {\mathcal O}})$, and self-adjoint in $L^2({\mathcal O}) \otimes \bhpa \cong L^2({\mathcal O};\bhpa)$.
     Due to \eqref{d5} and the non-negativity of $V_\omega$, we almost surely have
     $$
     \inf \sigma(H^{D}_{\omega, {\mathcal O}}) \geq E_1.
     $$
     For $E \in (-\infty, {\mathcal E})$,  and $\omega \in \Omega$ consider the quantity
     $$
    N(H^D_{\omega,{\mathcal O}}; E) : = {\rm Tr}\,\one_{(-\infty, E)}(H^{D}_{\omega, {\mathcal O}}),
    $$
    where, in accordance with our general notations, $\one_{(-\infty, E)}(H^{D}_{\omega,{\mathcal O}})$ is the spectral projection of the operator $H^{D}_{\omega, {\mathcal O}}$ corresponding to $(-\infty, E)$. Thus, ${\rm Tr}\,\one_{(-\infty, E)}(H^{D}_{\omega, {\mathcal O}})$  is the number of the eigenvalues of $H^{D}_{\omega,{\mathcal O}}$ smaller than $E$, and counted with their multiplicities.\\
    Pick $L \in (0,\infty)$, and set
    $C_L : = \left(-\frac{L}{2}, \frac{L}{2}\right)^d$, ${\mathcal Z}(C_L) : = \inf \sigma\left(H_{\perp, C_L}^D\right)$.\\

    In the sequel we will need the following simple
    \begin{lemma} \label{l1}
    The function $(0,\infty) \ni L \mapsto  {\mathcal Z}(C_L) \in (0,\infty)$ is decreasing, and
    \bel{d38}
    \lim_{L \to \infty} {\mathcal Z}(C_L) = 0.
    \ee
    \end{lemma}
    \begin{proof}
    If $B=0$, then  ${\mathcal Z}(C_L) = d \pi^2 L^{-2}$ which implies \eqref{d38}. If $B \neq 0$, then in $\re^d$ there exist Cartesian coordinates such that
    $$
    B = \left\{
    \begin{array} {l}
    \bigoplus_{j=1}^m \left(
    \begin{array} {cc}
    0 & b_j\\
    -b_j & 0
    \end{array}
    \right) \quad {\rm if} \quad n=0,\\
    \left(\bigoplus_{j=1}^m \left(
    \begin{array} {cc}
    0 & b_j\\
    -b_j & 0
    \end{array}
    \right)\right) \bigoplus {\mathbb O}_n \quad {\rm if} \quad n \geq 1,
    \end{array}
    \right.
    $$
    where ${\mathbb O}_n$ is the zero $n \times n$ matrix (see e.g. \cite{mohrai}). If the sides of the cube $\tilde{C}_L \subset \re^d$, centered at the origin, are parallel to the coordinate
    hyperplanes corresponding to this coordinate system, then we have
    \bel{d50}
    {\mathcal Z}(\tilde{C}_L) = \sum_{j=1}^m(\zeta_j(L) - b_j) + n\pi^2 L^{-2},
    \ee
    where $\zeta_j(L)$, $j=1,\ldots,m$, is the smallest eigenvalue of the self-adjoint operator generated in $L^2\left(S_L\right)$ with
    $S_L: = \left(-\frac{L}{2}, \frac{L}{2}\right)^2$
    by the closed non-negative quadratic form
    $$
    \int_{S_L} \left(\left|i\frac{\partial f}{\partial x_1} - \frac{b_j x_2}{2} f\right|^2 +
    \left|i\frac{\partial f}{\partial x_2} + \frac{b_j x_1}{2} f\right|^2\right) dx, \quad f \in {\rm H}_0^1\left(S_L\right).
    $$
    By \cite[Proposition 4.1]{erd}, we have $\zeta_j(L) > b_j$ if $L \in (0,\infty)$, and
    \bel{d51}
    \lim_{L \to \infty} \frac{\ln{(\zeta_j(L) - b_j)}}{L^2} = - \frac{b_j}{2\pi}, \quad j=1,\ldots,m.
    \ee
    Finally, for any cube $C_L$ centered at the origin, we have
    \bel{d52}
    {\mathcal Z}(\tilde{C}_{\sqrt{d} L}) \leq {\mathcal Z}(C_L) \leq {\mathcal Z}(\tilde{C}_{L/\sqrt{d}}), \quad L \in (0,\infty).
    \ee
    Now \eqref{d38} in the case $B \neq 0$ follows from \eqref{d50} - \eqref{d52}.
    \end{proof}
    \begin{theorem} \label{th1}
    Assume ${\bf H_1} - {\bf H_4}$. Then there exists a left-continuous non-decreasing function $\nu_V : (-\infty,{\mathcal E}) \to [0,\infty)$ and a set $\Omega_0 \in {\mathcal A}$ of full probability, i.e. ${\mathbb P}(\Omega_0) = 1$, such that for each $\omega \in \Omega_0$ we have
    \bel{d23}
    \lim_{L \to \infty} L^{-d} {\rm Tr}\,\one_{(-\infty, E)}(N(H^D_{\omega,C_L}; E)) = \nu_V(E)
    \ee
    at the continuity points $E \in (-\infty, {\mathcal E})$ of $\nu_V$.
    \end{theorem}
   {\em Remarks}: (i)  The function $\nu_V$ is called the integrated density of surface states (IDSS) for the operator  $H_{\omega}$.
   Since it is non-decreasing, the set of its discontinuity points is  countable. By definition, $\nu_V$ is non-random.  As mentioned in the Introduction, he IDSS for
   non-magnetic quantum Hamiltonians was first introduced in \cite{ekss} where its general properties were studied in detail. A further
   development of the theory of the IDSS can be found in \cite{kw}.\\
    (ii)
     Since we define the quadratic form \eqref{d6} on ${\rm H}_0^1({\mathcal O})$, it is natural to call $\nu_V$ the {\em Dirichlet} IDSS.
     Let us discuss briefly the possibility to introduce also a {\em Neumann} IDSS.
     Define the operator $H^{N}_{\perp, C_L}$ as the self-adjoint operator generated in $L^2(C_L)$ by the closed lower bounded quadratic form
    $$
    \int_{C_L} \left(|i\nabla f + Af|^2   - \beta)|f|^2\right) dx, \quad f \in {\rm H}^1(C_L).
    $$
    Again, the spectrum of $H^{N}_{\perp, C_L}$ is purely discrete. However, if $m>0$ and, for instance,  the sides of $C_L$ are parallel to the hyperplanes corresponding to the coordinate system described in the proof of Lemma \ref{l1}, then \cite[Theorem 1.2]{bon} easily implies
     $$
\inf \sigma(H^{N}_{\perp, C_L}) = (\Theta-1) \beta + O(L^{-1/2}), \quad L \to \infty,
$$
with a constant $\Theta \in (0,1)$ independent of $B$ and $L$ (see also \cite{hm} for a related result in the case where $C_L$ is replaced by a domain with a smooth boundary). Therefore, $\inf \sigma(H^{N}_{\perp, C_L}) < 0$ for $L$ large enough. Hence, if we assume that ${\mathcal E} < \infty$, and introduce the operators
$$
H^{N}_{0, C_L} : = H^{N}_{\perp, C_L} \otimes I_\parallel + I_{\perp, C_L} \otimes \hpa, \quad H^{N}_{\omega, C_L} : = H^{N}_{0, C_L} + V_{\omega, C_L},
$$
we find that $\inf \sigma_{\rm ess}(H^{N}_{0, C_L}) < {\mathcal E}$, and, generally speaking, we cannot rule out the possibility that $\inf \sigma_{\rm ess}(H^{N}_{\omega, C_L}) < {\mathcal E}$. In such a case,
$$
{\rm Tr}\,\one_{(-\infty, E)}(H^{N}_{\omega, C_L}) = \infty,
$$
if $E \in \left(\inf \sigma_{\rm ess}(H^{N}_{\omega, C_L}), {\mathcal E}\right)$, and the Neumann IDSS would not be well defined, at least not for all energies
$E \in (-\infty, {\mathcal E})$. That is why we do not consider it in the present article. Note, however, that if $m=0$, i.e. $B = 0$, then
$\inf \sigma(H^{N}_{\perp, L}) = 0$, the Neumann IDSS is correctly defined, and under generic assumptions it coincides with the Dirichlet IDSS
(see \cite{kw}). Also, if $m \geq 0$, and ${\mathcal E} = \infty$, the Neumann IDSS would be well defined. \\

\begin{proof}[Proof of Theorem \ref{th1}{\rm :}] We follow the general lines of the proof of \cite[Proposition 2.4]{kw}. Our goal is to check that
the stochastic process $N(H^{D}_{\omega,C_L};E)$ indexed by the cubes $C_L + \xi \subset \rd$ with
$L \in (0,\infty)$ and $\xi \in \re^d$ if ${\mathbb G} = \re$, or with
$L \in \N$ and $\xi \in \Z^d$ if ${\mathbb G} = \Z$, satisfies the hypotheses of the Akcoglu-Krengel theorem (see \cite{ak});
then we can argue as in the proof of \cite[Theorem 3.2]{km}. \\
First, if ${\mathcal O}_1 \cap {\mathcal O}_2 = \emptyset$, and ${\mathcal O} : = \left(\overline{{\mathcal O}_1 \cup {\mathcal O}_2}\right)^{\rm Int}$, then
$$
N(H^{D}_{\omega,{\mathcal O}};E) \geq N(H^{D}_{\omega,{\mathcal O}_1};E) + N(H^{D}_{\omega,{\mathcal O}_2};E).
$$

Define the family $\left\{\tau_\xi\right\}_{\xi \in \re^d}$ of magnetic translations by
    \bel{d11}
    (\tau_\xi f)(x) := \exp{\left(-\frac{i}{2} \sum_{j, k = 1}^d \xi_j B_{j k} x_k\right)} f(x-\xi),
    \quad \xi \in \re^d, \quad x \in \re^d, \quad f \in L^2(\re^d).
    \ee
    Thus, $\tau_\xi$, $\xi \in \re^d$, is a unitary operator in $\bhpe$. On $C^1(\re^d)$ we have
    \bel{d12}
\tau_\xi \left(-i\frac{\partial}{\partial x_j} - A_j\right) \tau_\xi^* = -i\frac{\partial}{\partial x_j} - A_j, \quad \xi \in \re^d, \quad j=1,\ldots,d.
    \ee
The restriction $\tau_{\xi, C_L}$ onto $L^2(C_L)$ of the magnetic translation  $\tau_{\xi}$ (see \eqref{d11}), is a unitary operator form $L^2(C_L)$ onto $L^2(C_L + \xi)$, and a bijection form ${\rm Dom}(H^{D}_{\perp,C_L})$ onto ${\rm Dom}(H^{D}_{\perp,C_L + \xi})$. Similarly, $\tau_{\xi, C_L} \otimes I_\parallel$ is a unitary operator form $L^2(C_L) \otimes \bhpa$ onto $L^2(C_L + \xi) \otimes \bhpa$, and a bijection form ${\rm Dom}(H^{D}_{0,C_L})$ onto ${\rm Dom}(H^{D}_{0,C_L + \xi})$.  By \eqref{d12} and ${\bf H_4}$, we have
    $$
\left(\tau_{\xi, C_L} \otimes I_\parallel\right) H_{\omega, C_L + \xi} \left(\tau_{\xi, C_L} \otimes I_\parallel\right)^* = H_{{\mathcal T}_\xi \omega, C_L}, \quad \xi \in {\mathbb G}^d.
    $$
    Therefore,
    $$
    N(H^{D}_{\omega,C_L + \xi};E)  = N(H^{D}_{{\mathcal T}_\xi \omega,C_L};E), \quad \xi \in {\mathbb G}^d.
    $$
It remains to check that
    \bel{5}
    \sup_{L \in (0,\infty)}L^{-d}{\mathbb E}(N(H^{D}_{\omega,C_L};E)) < \infty,
    \ee
    where ${\mathbb E}$ denotes the expectation with respect to the probability measure $d{\mathbb P}$. By the non-negativity of $V_\omega$ (see ${\bf H_3}$), and \eqref{d5}, we have almost surely
    \bel{d9}
    N(H^{D}_{\omega,C_L};E) \leq N(H^{D}_{0,C_L};E) = \sum_{j\in {\mathcal J} \, : \, E_j < E} {\rm Tr}\,\one_{(-\infty, E-E_j)}(H^D_{\perp, C_L}), \quad
    E \in (-\infty, {\mathcal E}).
     \ee
     Further, the minimax principle easily implies
     \bel{d15}
      {\rm Tr}\,\one_{(-\infty, E-E_j)}(H^D_{\perp, C_L}) \leq {\rm Tr}\,\one_{(-\infty, 0)}(H_{\perp} + \beta + 1 - \eta \one_{C_L})), \quad E_j < E.
      \ee
      with $\eta : = \beta + 1 + E - E_j$ and $E_j < E$.
   Next, for a compact linear operator $G$ in a separable Hilbert space, and for $s>0$, set
    $$
    n_*(s;G) : = {\rm Tr}\, \one_{(s^2,\infty)}(G^*G).
    $$
    Thus, $n_*(s;G)$ is the number of the singular values of $G$, greater than $s>0$, and counted with their multiplicities.
    Then the Birman-Schwinger principle (see e.g. \cite[Lemma 1.1]{bir}), implies
     \bel{8}
    {\rm Tr}\,\one_{(-\infty, 0)}(H_{\perp} + \beta + 1 - \eta \one_{C_L}) =
    n_*(\eta^{-1/2}; \one_{C_L} (H_\perp + \beta + 1)^{-1/2}).
    \ee
    Let $p > d$ be an even integer number. Then it follows from an elementary Chebyshev-type estimate, and the diamagnetic inequality (see e.g. \cite{ahs}), that
    $$
    n_*(\eta^{-1/2}; \one_{C_L} (H_\perp + \beta + 1)^{-1/2}) \leq \eta^{p/2} \| \one_{C_L} (H_\perp + \beta + 1)^{-1/2}\|_p^p \leq
    $$
    \bel{9}
    \eta^{p/2} \| \one_{C_L} (-\Delta  + 1)^{-1/2}\|_p^p,
    \ee
    where $\|G\|_p : = \left({\rm Tr}\,(G^*\,G)^{p/2}\right)^{1/p}$, $p \in [1,\infty)$, denotes the norm of the operator $G$ in the
    $p$th Schatten-von Neumann class.
    A standard interpolation result (see e.g. \cite[Theorem 4.1]{s}), implies
    \bel{10}
    \| \one_{C_L} (-\Delta  + 1)^{-1/2}\|_p^p \leq
    (2\pi)^{-d}  \int_{\re^d} (|\xi|^2 + 1)^{-p/2}d\xi\, L^d.
    \ee
   Now \eqref{5} follows from \eqref{d9} - \eqref{10}.

    \end{proof}

\section{Estimates of the IDSS}
\label{s3} \setcounter{equation}{0}
Introduce the function $\Omega \times \re^d \ni (\omega,x) \mapsto W_\omega(x) \in [0,\infty)$.
 In this section we define the  integrated density of bulk states ${\mathcal N}_W$ for a reduced operator $\hpe + W_\omega$ with
 certain $W_\omega$ related to $V_\omega$, and estimate the IDSS
 $\nu_V$  in terms of  ${\mathcal N}_W$.\\

    Assume that $W_\omega$ satisfies the hypotheses ${\bf H_1} - {\bf H_4}$ with
    $\bhpa = \C$. For $E \in \re$ set
    \bel{d17}
{\mathcal N}_W(E) : = {\mathbb E}\left({\rm Tr}\,\left(\one_{C_1} \one_{(-\infty,E)}(\hpe + W_\omega) \one_{C_1}\right)\right).
    \ee
Thus, ${\mathcal N}_W$ is the usual integrated density of states (IDS) for the random ${\mathbb G}^d$-ergodic operator $\hpe + W_\omega$.
Due to the ergodicity of $H_{0,\perp} + W_\omega$, there exists a set $\Sigma \subset \re$ such that almost surely
$$
\sigma(\hpe + W_\omega) = \Sigma,
$$
and
    \bel{sof2}
\Sigma = {\rm supp}\,d{\mathcal N}_W
    \ee
(see \cite{km2, pf}).
The IDS ${\mathcal N}_W$ admits a representation as a thermodynamic limit of normalized finite-volume eigenvalue counting functions:

\begin{theorem} \label{th3} {\rm \cite{dim},\cite[Theorem 3.1]{hlmw}} Assume that $W_\omega$ satisfies ${\bf H_1} -{\bf H_4}$ with $\bhpa = \C$.
Then almost surely
    \bel{d22}
\lim_{L \to \infty} L^{-d} {\rm Tr}\,\one_{(-\infty,E)}(H^D_{\perp,C_L} + W_\omega) = {\mathcal N}_W(E),
    \ee
at the points of continuity $E \in \re$ of ${\mathcal N}_W(E)$.
\end{theorem}
If $W_\omega=0$, then \eqref{d17} easily implies
    \bel{d16}
    {\mathcal N}_0(E) = \left\{
    \begin{array} {l}
    \frac{\omega_d}{(2\pi)^d} E_+^{d/2} \quad {\rm if} \quad m=0, d=n \geq 1,\\
    \frac{b_1\ldots b_m}{(2\pi)^m}\frac{\omega_n}{(2\pi)^n} \sum_{q=0}^\infty \mu_q (E-\Lambda_q)_+^{n/2} \quad {\rm if} \quad m \geq 1, n \geq 1,\\
    \frac{b_1\ldots b_m}{(2\pi)^m} \sum_{q=0}^\infty \mu_q \one_{(-\infty,E)}(\Lambda_q) \quad {\rm if} \quad m \geq 1, n = 0.
    \end{array}
    \right.
    \ee
    Here $\omega_d : = \frac{\pi^{d/2}}{\Gamma(d/2 + 1)}$, $\Gamma$ being the Euler gamma function, is the volume of the unit ball in $\re^d$, $d \geq 1$, and
  $$
\mu_q : = \#\left\{(l_1,\ldots,l_m) \in \Z_+^m \, | \, 2\sum_{j=1}^m b_j l_j = \Lambda_q\right\}, \quad q \in \Z_+,
$$
is the multiplicity of the Landau level $\Lambda_q$, $q \in \Z_+$. Note that if $n \geq 1$, then ${\mathcal N}_0$ is continuous on $\re$,
while if $n=0$, its discontinuity points are the Landau levels $\Lambda_q$, $q \in \Z_+$. Moreover it is easy to check that for any $d = 2m + n \geq 1$ we have
    \bel{d18a}
    \lim_{E \to \infty} E^{-d/2} {\mathcal N}_0(E) =  \frac{\omega_d}{(2\pi)^d};
    \ee
    in particular, the semi-classical asymptotic coefficient $\frac{\omega_d}{(2\pi)^d}$ is independent of the magnetic field $B$. \\
    Further, denote by $\rho : (-\infty, {\mathcal E}) \to \Z_+$ the eigenvalue counting function for the operator $\hpa$, i.e.
    $$
    \rho(E) : = {\rm Tr}\,\one_{(-\infty,E)}(\hpa), \quad E \in (-\infty,{\mathcal E}).
    $$
    Set
    \bel{d18}
    ({\mathcal N}_0 * d\rho)(E) := \sum_{j \in {\mathcal J}}{\mathcal N}_0(E-E_j), \quad E \in (-\infty,{\mathcal E}).
    \ee
 Note that ${\mathcal N}_0(E) \neq 0$ if and only if $E>0$; therefore, only the terms in \eqref{d18} which correspond to eigenvalues $E_j < E$ do not vanish.
 Since $E < {\mathcal E}$, the non-vanishing terms in \eqref{d18} are finitely many.
 \begin{pr} \label{pr1} Assume that $V_\omega$ satisfies ${\bf H_1} - {\bf H_4}$. Then we have
    \bel{d19}
    ({\mathcal N}_0 * d\rho)(E-M) \leq \nu_V(E) \leq ({\mathcal N}_0 * d\rho)(E), \quad E \in (-\infty, {\mathcal E}).
    \ee
    \end{pr}
    \begin{proof}
    The mini-max principle and hypotheses ${\bf H_2} -{\bf H_3}$ easily imply that almost surely
    \bel{d20}
    N(H^D_{0,C_L}; E-M) \leq N(H^D_{\omega,C_L}; E) \leq N(H^D_{0,C_L}; E), \quad E \in (-\infty, {\mathcal E}).
    \ee
    On the other hand,
    \bel{d21}
    N(H^D_{0,C_L}; E) = \sum_{j \in {\mathcal J} \, : \, E_j<E} \one_{(-\infty, E_j-E)}(H^D_{\perp,C_L}), \quad E \in (-\infty, {\mathcal E}).
    \ee
    Now if $E \in (-\infty, {\mathcal E})$ is a common continuity point of the functions ${\mathcal N}_0 * d\rho$, $\nu_V$, and
    $({\mathcal N}_0 * d\rho)(\cdot - M)$, then \eqref{d19} follows from \eqref{d20} - \eqref{d21}, combined with \eqref{d23}, \eqref{d22}, and \eqref{d18}.
    In order to prove \eqref{d19} for general $E \in (-\infty, {\mathcal E})$, we apply an appropriate limiting argument, taking into account that the three functions
    ${\mathcal N}_0 * d\rho$, $\nu_V$, and
    $({\mathcal N}_0 * d\rho)(\cdot - M)$ are left continuous and non-decreasing so that the set of their discontinuity points is countable.
    \end{proof}
    As an immediate application of Proposition \ref{pr1}, we have the following
    \begin{follow} \label{f1}
    Assume that $V_\omega$ satisfies ${\bf H_1} - {\bf H_4}$. Let ${\mathcal E} = \infty$. Suppose that there exist constants
    $\theta \in (0,\infty)$ and $C \in (0,\infty)$, such that
    \bel{d24}
    \rho(E) = C E^\theta (1 + o(1)), \quad E \to \infty.
    \ee
    Then we have
    \bel{d25}
    \nu_V(E) = C \frac{d\theta}{d+2\theta} {\rm B}(d/2,\theta) \frac{\omega_d}{(2\pi)^d} E^{\frac{d}{2} + \theta}(1 + o(1)), \quad E \to \infty,
    \ee
    where ${\rm B}$ is the Euler beta function.
    \end{follow}
    \begin{proof}
    Asymptotic relation \eqref{d25} follows easily from \eqref{d19}, \eqref{d18a}, \eqref{d24}, and the Karamata Tauberian theorem (see the original work \cite{kar} or \cite[Problem 14.2]{shu}).
    \end{proof}
    Our next goal is to estimate $\nu_V$ for energies $E$ close to the lower edges of the bands of ${\rm supp}\,d\nu_V$, i.e. close to the upper edges of the gaps in ${\rm supp}\,d\nu_V$. First, we estimate $\nu_V$ for energies $E$ close to $E_1 = \inf \, {\rm supp}\,d\nu_V$.
    Note that \eqref{d19} implies that $\inf {\rm supp}\,d\nu_V \geq E_1$.
    Assume that $E_1$ is a simple eigenvalue of $\hpa$. Set
    $$
    W_{1,\omega}(x) : = \langle V_\omega(x) \psi_1, \psi_1\rangle_{\bhpa}, \quad x \in \re^d, \quad \omega \in \Omega.
    $$
    Evidently, if $V_\omega$ satisfies hypotheses ${\bf H_1} - {\bf H_4}$ with arbitrary separable Hilbert space $\bhpa$,
    then $W_{1,\omega}$ meets these conditions for $\bhpa = \C$.
    \begin{theorem} \label{th5}
     Assume $V_\omega$ satisfies hypotheses ${\bf H_1} - {\bf H_4}$, and that $E_1$ is a simple eigenvalue of $\hpa$.
     Let $\lambda_* \in (0,E_2-E_1)$ and $\delta \in \left(\frac{M}{M+ E_2-E_1-\lambda_*},1\right)$. Then we have
     \bel{d26}
     {\mathcal N}_{W_1}(\lambda) \leq \nu_V(E_1 + \lambda) \leq {\mathcal N}_{(1-\delta)W_1}(\lambda), \quad \lambda \in (0,\lambda_*].
     \ee
    \end{theorem}
 \begin{proof}
    Introduce the orthogonal projection $P_1: L^2(C_L;\bhpa) \to L^2(C_L;\bhpa)$ by
$$
(P_1f)(x) = \langle f(x), \psi_1 \rangle_{\bhpa} \psi_1 \quad x \in C_L, \quad f \in L^2(C_L; \bhpa).
$$
 Set $Q_1 : = I - P_1$, and
$$
{\mathcal D}_1 : = P_1\,{\rm Dom}\,(H^{D}_{0,C_L}), \quad {\mathcal C}_1 : = Q_1\,{\rm Dom}\,(H^{D}_{0,C_L}).
$$
It is easy to see that ${\mathcal D}_1 \subset {\rm Dom}\,(H^{D}_{0,C_L}) = {\rm Dom}\,(H^{D}_{\omega,C_L})$ and ${\mathcal C}_1 \subset
{\rm Dom}\,(H^{D}_{\omega,C_L})$. We will consider  $\left(P_1 H^{D}_{\omega,C_L} P_1\right)_{|{\mathcal D}_1}$
(resp., $\left(Q_1 H^{D}_{\omega,C_L} Q_1\right)_{|{\mathcal C}_1}$) as a self-adjoint operator in the Hilbert space $P_1 L^2(C_L; \bhpa)$ (resp., $Q_1 L^2(C_L; \bhpa)$).
Now note that
 the operator $\left(P_1 H^{D}_{\omega,C_L} P_1\right)_{|{\mathcal D}_1}$ is unitarily equivalent to the operator $H_{\perp, C_L}^{D} + W_{1,\omega} + E_1$.
 More precisely,
    \bel{14}
{\mathcal U}^* \left(\left(P_1 H^{D}_{\omega,C_L} P_1\right)_{|{\mathcal D}_1}\right) {\mathcal U} = H_{\perp, C_L}^{D} + W_{1,\omega} + E_1,
    \ee
where ${\mathcal U} : L^2(C_L) \to P_1 L^2(C_L; \bhpa)$ is the unitary operator defined by
$$
({\mathcal U} g)(x) : = g(x) \psi_1, \quad x \in C_L,  \quad g \in L^2(C_L).
$$
Moreover, we have
    \bel{15}
    \inf \sigma\left(\left(Q_1 H^{D}_{0,C_L} Q_1\right)_{|{\mathcal C}_1}\right) = E_2 + {\mathcal Z}(C_L).
    \ee

Let $\lambda \in (0,\lambda_*]$ with $\lambda_* \in (0,E_2-E_1)$. The mini-max principle and \eqref{14} entail
$$
N(H^{D}_{\omega, C_L}; E_1 + \lambda) \geq
$$
    \bel{16}
    {\rm Tr}\,\one_{(-\infty, E_1+\lambda)}\left(\left(P_1 H^{D}_{\omega, C_L} P_1\right)_{|{\mathcal D}_1}\right) =
    {\rm Tr}\,\one_{(-\infty, \lambda)}(H_{\perp, C_L}^{D} + W_{1,\omega}).
    \ee
    Pick $\delta \in \left(\frac{M}{M+ E_2-E_1-\lambda_*},1\right)$. Then the operator inequality
    $$
    H^{D}_{\omega, C_L} =
    $$
    $$
    P_1\left(H^{D}_{0, C_L} + V_\omega\right)P_1 + Q_1\left(H^{D}_{0, C_L} + V_\omega\right)Q_1 + 2 {\rm Re}\, P_1 V_\omega Q_1 \geq
    $$
    $$
    P_1\left(H^{D}_{0, C_L} + (1-\delta)V_\omega\right)P_1 + Q_1\left(H^{D}_{0, C_L} + (1-\delta^{-1})V_\omega\right)Q_1,
    $$
    combined with the mini-max principle and \eqref{14}, implies
    $$
    N(H^{D}_{\omega, C_L}; E_1 + \lambda) \leq
$$
$$
{\rm Tr}\,\one_{(-\infty, E_1+\lambda)}\left(\left(P_1 \left(H^{D}_{0, C_L} + (1-\delta)V_\omega\right)P_1\right)_{|{\mathcal D}_1}\right) +
$$
$$
{\rm Tr}\,\one_{(-\infty, E_1+\lambda)}\left(\left(Q_1 \left(H^{D}_{0, C_L} + (1-\delta^{-1})V_\omega\right)Q_1\right)_{|{\mathcal C}_1}\right) \leq
$$
    \bel{17}
    {\rm Tr}\,\one_{(-\infty, \lambda)}\left(H_{\perp,C_L}^{D} + (1-\delta)W_{1,\omega}\right) +
    {\rm Tr}\,\one_{(-\infty, E_1+\lambda + (\delta^{-1}-1)M)}\left(\left(Q_1 H^{D}_{0,C_L}Q_1\right)_{|{\mathcal C}_1}\right).
    \ee
    Now note that our choice of $\lambda$ and $\delta$ implies
    $E_1+\lambda + (\delta^{-1}-1)M < E_2$.
    Therefore, by \eqref{15}, we have
    \bel{18}
    {\rm Tr}\,\one_{(-\infty, E_1+\lambda + (\delta^{-1}-1)M)}\left(\left(Q_1 H^{D}_{0,C_L}Q_1\right)_{|{\mathcal C}_1}\right) = 0, \quad \lambda \in (0,\lambda_*].
    \ee
    Now, the lower bound in \eqref{d26} follows from \eqref{16} while the upper bound follows form
    \eqref{17} - \eqref{18} combined with \eqref{d22} and \eqref{d23}.
    \end{proof}

     Our next goal is to estimate the IDSS $\nu_V$ near energies which play the role of upper edges of internal gaps of ${\rm supp}\,d\nu_V$.
     Assume that $n = 0$ and $E_1 + \Lambda_1 > {\mathcal E}$. Then by \eqref{d28} we have
     $$
     \sigma(H_0) \cap (-\infty,{\mathcal E}) = \bigcup_{j \in {\mathcal J}} \left\{E_j\right\},
     $$
     and the energies $E_j$ are eigenvalues of $H_0$ of infinite multiplicity. Assume that $r \geq 2$, and there exists $j \in {\mathcal J}$, $j \geq 2$, such that
     \bel{d30}
     E_{j-1} < E_j < E_{j+1}.
     \ee
      Moreover, assume that
      \bel{d31}
      M < E_j - E_{j-1}.
      \ee
       By \eqref{d19}, \eqref{d18}, and \eqref{d16} with $n=0$, the IDSS $\nu_V$
     is constant on the interval $[E_j-M, E_j]$. More precisely,
     \bel{d29}
     \nu_V(E) = (j-1) \frac{b_1 \ldots b_m}{(2\pi)^m}, \quad   E \in [E_j-M, E_j].
     \ee
     Thus, we are going to estimate the difference $\nu_V(E_j+\lambda)-\nu_V(E_j)$ for $\lambda > 0$ small enough. Set
    $$
    W_{j,\omega}(x) : = \langle V_\omega(x) \psi_j, \psi_j\rangle_{\bhpa}, \quad x \in \re^d, \quad \omega \in \Omega.
    $$

     \begin{theorem} \label{th6}
     Assume $V_\omega$ satisfies hypotheses ${\bf H_1} - {\bf H_4}$, $r \geq 2$, and there exists $j \in {\mathcal J}$, $j \geq 2$, such that
     \eqref{d30} and \eqref{d31} hold true.
     Let $\delta_- \in \left(\frac{M}{E_j-E_{j-1}-M},\infty\right)$, $\lambda_* \in \left(0,\min\left\{E_{j+1}-E_j, (1+\delta_-^{-1})M\right\}\right)$, and
     $\delta_+ \in \left(\frac{M}{M+ E_{j+1}-E_j-\lambda_*},1\right)$. Then we have
     \bel{d32}
     {\mathcal N}_{(1+\delta_-)W_j}(\lambda) \leq \nu_V(E_j + \lambda) - \nu_V(E_j) \leq {\mathcal N}_{(1-\delta_+)W_j}(\lambda), \quad \lambda \in (0,\lambda_*].
     \ee
    \end{theorem}
 \begin{proof}
 The proof of \eqref{d32} is similar to the one of \eqref{d26}, so that we omit certain details.
    Introduce the orthogonal projection $P_j: L^2(C_L;\bhpa) \to L^2(C_L;\bhpa)$ by
$$
(P_j f)(x) = \langle f(x), \psi_j \rangle_{\bhpa} \psi_j \quad x \in C_L, \quad f \in L^2(C_L; \bhpa).
$$
 Set $Q_j : = I - P_j$, and
$$
{\mathcal D}_j : = P_j\,{\rm Dom}\,(H^{D}_{0,C_L}), \quad {\mathcal C}_j : = Q_j\,{\rm Dom}\,(H^{D}_{0,C_L}).
$$

The operator $\left(P_j H^{D}_{\omega,C_L} P_j\right)_{|{\mathcal D}_j}$ is unitarily equivalent to the operator $H_{\perp, C_L}^{D} + W_{j,\omega} + E_j$.
 Moreover, we have
    \bel{d34}
    \sigma\left(\left(Q_j H^{D}_{0,C_L} Q_j\right)_{|{\mathcal C}_j}\right) \cap (E_{j-1} + {\mathcal Z}(C_L), E_{j+1}) = \emptyset.
    \ee

Let us first prove the lower bound in \eqref{d32}. Bearing in mind the operator inequality
$$
H^{D}_{\omega, C_L} \leq P_j\left(H^{D}_{0, C_L} + (1+\delta_-)V_\omega\right)P_j + Q_j\left(H^{D}_{0, C_L} + (1+\delta_-^{-1})V_\omega\right)Q_j,
$$
 we find that the mini-max principle and the unitary equivalence of the operators $\left(P_j (H^{D}_{0,C_L} + (1+\delta_-)V_\omega)P_j\right)_{|{\mathcal D}_j}$
 and $H_{\perp, C_L}^{D} + (1+\delta_-)W_{j,\omega} + E_j$, entail
$$
N(H^{D}_{\omega, C_L};E_j + \lambda)  \geq
$$
    \bel{d35}
    {\rm Tr}\,\one_{(-\infty, \lambda)}(H_{\perp, C_L}^{D} + (1+\delta_-)W_{1,\omega}) +
    {\rm Tr}\,\one_{(-\infty, E_j+\lambda -(1+\delta_-^{-1}) M)}\left(\left(Q_j \left(H^{D}_{0, C_L}\right)Q_j\right)_{|{\mathcal C}_j}\right).
\ee
On the other hand, by  the non-negativity of $V_\omega$, and
    $$
    \inf \sigma\left(\left(P_j H^{D}_{0,C_L} P_j\right)_{|{\mathcal D}_j}\right)  = E_j + {\mathcal Z}(C_L) > E_j,
    $$
    we find that
    \bel{d39}
N(H^{D}_{\omega, C_L};E_j) \leq {\rm Tr}\,\one_{(-\infty, E_j)}\left(H^{D}_{0, C_L}\right) =
{\rm Tr}\,\one_{(-\infty, E_j)}\left(\left(Q_j H^{D}_{0,C_L} Q_j\right)_{|{\mathcal C}_j}\right).
    \ee
    Combining \eqref{d35} and \eqref{d39}, we get
    $$
    N(H^{D}_{\omega, C_L};E_j + \lambda) - N(H^{D}_{\omega, C_L};E_j) \geq
    $$
    \bel{d40}
    {\rm Tr}\,\one_{(-\infty, \lambda)}(H_{\perp, C_L}^{D} + (1+\delta_-)W_{1,\omega}) -
    {\rm Tr}\,\one_{[E_j+\lambda -(1+\delta_-^{-1}) M,E_j)}\left(\left(Q_j \left(H^{D}_{0, C_L}\right)Q_j\right)_{|{\mathcal C}_j}\right).
    \ee
    By our choice of $\delta_-$ and $\lambda_*$, and by \eqref{d31}, \eqref{d34}, and Lemma \ref{l1}, we find that there exists $L^-_0 \in (0, \infty)$ independent of $\lambda$ such
    that
    \bel{d41}
    {\rm Tr}\,\one_{[E_j+\lambda -(1+\delta_-^{-1}) M),E_j)}\left(\left(Q_j \left(H^{D}_{0, C_L}\right)Q_j\right)_{|{\mathcal C}_j}\right) = 0,
    \ee
    provided that $\lambda \in (0,\lambda_*)$, and $L \in (L^-_0,\infty)$. Now the lower bound in \eqref{d32} follows from \eqref{d40} and \eqref{d41},
    combined with \eqref{d22} and \eqref{d23}. \\
    Let us now prove the upper bound in \eqref{d32}. Using the operator inequality
$$
H^{D}_{\omega, C_L} \geq P_j\left(H^{D}_{0, C_L} + (1-\delta_+)V_\omega\right)P_j + Q_j\left(H^{D}_{0, C_L} + (1-\delta_+^{-1})V_\omega\right)Q_j,
$$
 we find that the mini-max principle and the unitary equivalence of the operators $\left(P_j (H^{D}_{0,C_L} + (1-\delta_+)V_\omega)P_j\right)_{|{\mathcal D}_j}$
 and $H_{\perp, C_L}^{D} + (1-\delta_+)W_{j,\omega} + E_j$, entail
$$
N(H^{D}_{\omega, C_L};E_j + \lambda)  \leq
$$
    \bel{d42}
    {\rm Tr}\,\one_{(-\infty, \lambda)}(H_{\perp, C_L}^{D} + (1-\delta_+)W_{1,\omega}) +
    {\rm Tr}\,\one_{(-\infty, E_j+\lambda +(\delta_+^{-1}-1) M)}\left(\left(Q_j \left(H^{D}_{0, C_L}\right)Q_j\right)_{|{\mathcal C}_j}\right).
\ee
On the other hand, the mini-max principle implies
    \bel{d43}
N(H^{D}_{\omega, C_L};E_j) \geq {\rm Tr}\,\one_{(-\infty, E_j)}\left(\left(Q_j H^{D}_{\omega,C_L} Q_j\right)_{|{\mathcal C}_j}\right) \geq
{\rm Tr}\,\one_{(-\infty, E_j-M)}\left(\left(Q_j H^{D}_{0,C_L} Q_j\right)_{|{\mathcal C}_j}\right).
    \ee
    Combining \eqref{d43} and \eqref{d44}, we get
    $$
    N(H^{D}_{\omega, C_L};E_j + \lambda) - N(H^{D}_{\omega, C_L};E_j) \leq
    $$
    \bel{d44}
    {\rm Tr}\,\one_{(-\infty, \lambda)}(H_{\perp, C_L}^{D} + (1-\delta_+)W_{1,\omega}) +
    {\rm Tr}\,\one_{[E_j-M,E_j+\lambda +(\delta_+^{-1}-1) M))}\left(\left(Q_j \left(H^{D}_{0, C_L}\right)Q_j\right)_{|{\mathcal C}_j}\right).
    \ee
    By our choice of $\delta_+$ and $\lambda_*$, \eqref{d34}, and Lemma \ref{l1}, we find that there exists $L^+_0 \in (0, \infty)$ independent of $\lambda$ such
    that
    \bel{d45}
    {\rm Tr}\,\one_{[E_j-M,E_j+\lambda +(\delta_+^{-1}-1) M))}\left(\left(Q_j \left(H^{D}_{0, C_L}\right)Q_j\right)_{|{\mathcal C}_j}\right) = 0,
    \ee
    provided that $\lambda \in (0,\lambda_*)$, and $L \in (L_0^+,\infty)$. Now the upper bound in \eqref{d32} follows from \eqref{d43} and \eqref{d44},
    combined with \eqref{d22} and \eqref{d23}. \\
    \end{proof}

\section{Applications}
\label{s4} \setcounter{equation}{0}
The applications of Theorem \ref{th5} (see Theorem \ref{th2} with $j =1$, and Theorem \ref{th7} below), concern the asymptotic behavior of $\nu_V(E)$  as $E \downarrow E_1$.
    As discussed in the Introduction, this behavior is characterized by a very rapid decay which
    usually goes under the name {\em Lifshits tail}. The application of Theorem \ref{th6} (see Theorem \ref{th2} with $j\geq 2$ below) deals with the internal Lifshits tails, i.e. with
    the asymptotic behavior of $\nu_V(E)$ as $E \downarrow E_j$ with $j \geq 2$, provided that \eqref{d30} - \eqref{d31} hold true.\\
    Assume that $V\omega$ is as in \eqref{2}. For $x \in \re^{d}$ and $j \in {\mathcal J}$ define the functions
    $$
    w_j(x) : = \int_{\re^{\ell}} v(x, y) \psi_j(y)^2 dy,
    $$
    and
    \bel{4}
    W_{j,\omega}(x) : = \sum_{\xi \in \Z^{d}} \lambda_\xi(\omega) w_j(x - \xi),
    \ee
    the one-site potential $v$ and the i.i.d. random variables $\{\lambda_\xi(\omega)\}_{\xi \in \Z^{d}}$ being the same as in \eqref{2}.
    Let $F(E) : = {\mathbb P}(\left\{\omega \in \Omega \, | \, \lambda_0(\omega) < E\right\})$, $E \in \re$. We
    suppose  that there exist  $E_0 \in (0,\infty)$ and $\kappa > 0$ such that
     \bel{d53}
     {\rm supp}\, F = [0,E_0],
     \ee
     and
    \bel{d54}
    F(E) \asymp E^\kappa, \quad E \downarrow 0.
    \ee
    {\em Remark}: In many particular cases, assumptions \eqref{d53} - \eqref{d54} could be relaxed. We state them here in a form which sometimes is too restrictive,
    just for the sake of the simplicity of exposition.
    \subsection{Surface Lifshits tails for magnetic quantum Hamiltonians}
    \label{ss41}
    In this subsection we assume that the unperturbed Hamiltonian $H_0$ is as in \eqref{d10}; in particular, $m = 1$, $n=0$, and $d = 2$, $\ell = 1$.
    We recall that in this case the discrete spectrum of $\hpa$ is simple.

    \begin{theorem} \label{th2}
     Let $j \in {\mathcal J}$. If $j \geq 2$ assume that \eqref{d31} holds true. Suppose that $V_\omega$ is of form \eqref{2} and satisfies ${\bf H_1} - {\bf H_3}$.\\
    {\rm (i)} Assume that
    $$
    c_- (1 + |x|)^{-\varkappa} \leq w_j(x) \leq c_+ (1 + |x|)^{-\varkappa}, \quad x \in \rd,
    $$
    for some $\varkappa > 2$, and $c_+ \geq c_- > 0$. Then we have
    $$
    \lim_{\lambda \downarrow 0} \frac{\ln{|\ln{\nu_V(E_j + \lambda) - \nu_V(E_j)|}}}{\ln{\lambda}} = - \frac{2}{\varkappa - 2}.
    $$
    {\rm (ii)} Assume that
    $$
    \frac{e^{-c_+|x|^\beta}}{c_+} \leq w_(x) \leq \frac{e^{-c_-|x|^\beta}}{c_-}, \quad x \in \rd,
    $$
    for some $\beta \in (0,2]$, and $c_+ \geq c_- > 0$. Then we have
    $$
    \lim_{\lambda \downarrow 0} \frac{\ln{|\ln{\nu_V(E_j + \lambda)- \nu_V(E_j)|}}}{\ln{|\ln{\lambda}|}} = 1 + \frac{2}{\beta}.
    $$
    {\rm (iii)} Assume that
    $$
    \frac{\one_S(x)}{c_+} \leq w_j(x) \leq \frac{e^{-c_-|x|^2}}{c_-}, \quad x \in \rd,
    $$
    for an open non-empty set $S \subset \rd$, and $c_+ \geq c_- > 0$. Then we have
    $$
    \lim_{\lambda \downarrow 0} \frac{\ln{|\ln{\nu_V(E_j + \lambda)- \nu_V(E_j)|}}}{\ln{|\ln{\lambda}|}} = 2.
    $$
    \end{theorem}
    Theorem \ref{th2} follows immediately from Theorems \ref{th5} - \ref{th6}, and the following result concerning
    the Lifshits tails for the IDS ${\mathcal N}_{W_j}$:

\begin{theorem} \label{th4} {\rm \cite{klorai, klopp} }
{\rm (i)} Under the assumptions of Theorem \ref{th2} (i) we have
    $$
    \lim_{\lambda \downarrow 0} \frac{\ln{|\ln{{\mathcal N}_{W_j}(\lambda)|}}}{\ln{\lambda}} = - \frac{2}{\varkappa - 2}.
    $$
    {\rm (ii)} Under the assumptions of Theorem \ref{th2} (ii) we have
    $$
    \lim_{\lambda \downarrow 0} \frac{\ln{|\ln{{\mathcal N}_{W_j}(\lambda)|}}}{\ln{|\ln{\lambda}|}} = 1 + \frac{2}{\beta}.
    $$
    {\rm (iii)}Under the assumptions of Theorem \ref{th2} (iii) we have
    $$
    \lim_{\lambda \downarrow 0} \frac{\ln{|\ln{{\mathcal N}_{W_j}(\lambda)|}}}{\ln{|\ln{\lambda}|}} = 2.
    $$
\end{theorem}

\subsection{Lifshits tails for non-magnetic quantum Hamiltonians}
    \label{ss42}
In this subsection we assume that the unperturbed Hamiltonian $H_0$ is as in \eqref{d51}; in particular, $B = 0$ and $d, \ell \in \N$ are arbitrary.
    In this case the lowest eigenvalue $E_1$ of $\hpa$ is simple. \\

    \begin{theorem} \label{th7}
     Suppose that $V_\omega$ is of form \eqref{2} and satisfies ${\bf H_1} - {\bf H_3}$.\\
    {\rm (i)} Assume that
    $$
    c_- (1 + |x|)^{-\varkappa} \leq w_1(x) \leq c_+ (1 + |x|)^{-\varkappa}, \quad x \in \re^d,
    $$
    for some $\varkappa \in (d,d+2)$, and $c_+ \geq c_- > 0$. Then we have
    $$
    \lim_{\lambda \downarrow 0} \frac{\ln{|\ln{\nu_V(E_1 + \lambda)|}}}{\ln{\lambda}} = - \frac{d}{\varkappa - d}.
    $$
    {\rm (ii)} Assume that
    $$
    c_- \one_S(x) \leq  w_1(x) \leq c_+ (1 + |x|)^{-d - 2}, \quad x \in \re^d,
    $$
    for some open non-empty set $S \subset \re^{d}$, and $c_\pm > 0$. Then we have
    $$
    \lim_{\lambda \downarrow 0} \frac{\ln{|\ln{\nu_V(E_1 + \lambda)|}}}{\ln{\lambda}} = - \frac{d}{2}.
    $$
    \end{theorem}
    Theorem \ref{th7} follows immediately from Theorem \ref{th5}, and the following, nowadays classical, results concerning
    the Lifshits tails for the IDS ${\mathcal N}_{W_j}$:

\begin{theorem} \label{th8}
{\rm (i)} {\rm \cite{km1} } Under the assumptions of Theorem \ref{th7} (i) we have
$$
    \lim_{\lambda \downarrow 0} \frac{\ln{|\ln{{\mathcal N}_{W_1}(\lambda)|}}}{\ln{\lambda}} = - \frac{d}{\varkappa - d}.
    $$
    {\rm (ii)} {\rm \cite{ks} } Under the assumptions of Theorem \ref{th7} (ii) we have
    $$
    \lim_{\lambda \downarrow 0} \frac{\ln{|\ln{{\mathcal N}_{W_1}(\lambda)|}}}{\ln{\lambda}} = - \frac{d}{2}.
    $$
\end{theorem}

{\em Remark}: The first (resp., the second) part of Theorem \ref{th7} recovers in our particular setting the result of \cite[Theorem 1.5]{kw}
(resp., \cite[Theorem 1.4]{kw}).\\

{\bf Acknowledgements}.
The major part of this work has been done during the second named author's visits to the University of Hagen,
Germany, in December 2014 - January 2015 and in February 2016, and to the Isaac Newton Institute, Cambridge, UK, in January-February 2015.
He thanks these institutions for financial support and hospitality. Both authors gratefully acknowledge the partial
support of the Chilean Scientific Foundation {\em Fondecyt}
under Grant 1130591, and of {\em N\'ucleo Milenio de F\'isica Matem\'atica} RC120002. The first named author gratefully acknowledges the hospitality extended to
him during his visits to the Pontificia Universidad Cat\'olica de Chile in 2014 and 2015.\\

{\sc Werner Kirsch}\\
Fakult\"at f\"ur Mathematik und Informatik\\
FernUniversit\"at in Hagen\\
Universit\"atsstrasse 1\\
D-58097 Hagen, Germany\\
E-mail: werner.kirsch@fernuni-hagen.de\\

{\sc Georgi Raikov}\\
Facultad de Matem\'aticas\\
Pontificia Universidad Cat\'olica de Chile\\
Av. Vicu\~na Mackenna 4860\\ Santiago de Chile\\
E-mail: graikov@mat.uc.cl


\begin{thebibliography} {[10]}
\frenchspacing \baselineskip=12 pt plus 1pt minus 1pt

\bibitem{ak} {\sc M. A. Akcoglu, U. Krengel}, {\em Ergodic theorems for superadditive processes}, J. Reine Angew. Math. {\bf 323} (1981), 53–-67.

\bibitem{aghkh} {\sc S. Albeverio, F. Gesztesy, R. H{\o}egh-Krohn, H. Holden}, {\em Solvable Models in Quantum Mechanics}, Second edition.  AMS Chelsea Publishing, Providence, RI, 2005.

\bibitem{abmg} {\sc W. O. Amrein, A. Boutet de Monvel, V. Georgescu}, {\em $C_0$-groups, commutator methods and spectral theory of
$N$-body Hamiltonians}, Progress in Mathematics, {\bf 135} Birkh\"auser Verlag, Basel, 1996.

\bibitem{abbfr} {\sc M.A.Astaburuaga, Ph.Briet, V.Bruneau, C.Fern\'andez,
G.D.Raikov}, {\em Dynamical resonances and SSF singularities for a
magnetic Schr\"odinger operator}, Serdica Mathematical Journal, {\bf
34} (2008), 179-218.

\bibitem{ahs} {\sc J. Avron, I. Herbst, B. Simon},
{\em Schr\"{o}dinger operators with  magnetic  fields.  I. General
interactions}, Duke Math. J. {\bf 45} (1978), 847-883.



\bibitem{bir}{\sc M. Sh. Birman}, {\em On the spectrum of singular boundary-value problems},  Mat. Sb. (N.S.) {\bf 55} (1961), 125–-174 (Russian);
 English translation in: Eleven Papers on Analysis, AMS Transl. {\bf 53}, 23-80, AMS, Providence, R.I., 1966.

\bibitem{bon}{\sc V. Bonnaillie}, {\em On the fundamental state energy for a Schr\"odinger operator with magnetic field in domains with corners}, Asymptot. Anal. {\bf 41} (2005),  215–-258.

\bibitem{bhkl} {\sc K. Broderix, D. Hundertmark, W. Kirsch, H. Leschke}, {\em The Fate of Lifshits Tails in Magnetic Fields}, J. Stat. Phys. {\bf 80} (1995) 1--22.

 \bibitem{dim} {\sc S. Doi, A. Iwatsuka, T. Mine}, {\em The uniqueness of the integrated density of states for the Schr\"odinger operators with magnetic fields}, Math. Z. {\bf 237} (2001),  335–-371.

\bibitem{ekss} {\sc H. Englisch, W. Kirsch, M. Schr\"oder, B. Simon},
{\em Random Hamiltonians ergodic in all but one direction},
Comm. Math. Phys. {\bf 128} (1990),  613–-625.

\bibitem{ekss1} {\sc H. Englisch, W. Kirsch, M. Schr\"oder, B. Simon},
{\em Density of surface states in discrete models}, Phys. Rev. Lett. {\bf 61} (1988), 1261--1262

\bibitem{erd} {\sc L. Erd\H{o}s}, {\em Lifschitz tail in a magnetic field: the nonclassical regime}, Probab. Theory Related Fields {\bf 112} (1998), 321–-371.

\bibitem{erd1} {\sc L. Erd\H{o}s}, {\em Lifschitz tail in a magnetic field: coexistence of classical and quantum behavior in the borderline case},
Probab. Theory Related Fields {\bf 121} (2001), 291–-236.

\bibitem{f} {\sc V. Fock}, {\em Bemerkung zur Quantelung des harmonischen Oszillators im
Magnetfeld},  Z. Physik {\bf 47} (1928), 446--448.

\bibitem{hm} {\sc B. Helffer, A. Morame}, {\em
Magnetic bottles in connection with superconductivity},
J. Funct. Anal. {\bf 185} (2001),  604-–680.


\bibitem{hlmw} {\sc T. Hupfer, H. Leschke, P. M\"uller, S. Warzel}, {\em Existence and uniqueness of the integrated density of states for Schr\"odinger operators with magnetic fields and unbounded random potentials}, Rev. Math. Phys. {\bf 13} (2001),  1547–-1581.

    \bibitem{kar} {\sc J. Karamata}, {\em Neuer Beweis und Verallgemeinerung der Tauberschen S\"{a}tze, welche die Laplacesche und Stieltjessche Transformation betreffen}, J. Reine Angew. Math. {\bf 164} (1931), 27–-39.

 \bibitem{km2} {\sc W. Kirsch,     F. Martinelli}, {\em  On the ergodic properties of the spectrum of general random operators}, J. Reine Angew. Math. {\bf 334}  (1982), 141–-156.

\bibitem{km0} {\sc W. Kirsch,     F. Martinelli}, {\em  On the spectrum of Schr\"odinger operators with a random potential}, Comm. Math. Phys. {\bf 85} (1982), 329–-350.

\bibitem{km} {\sc W. Kirsch, F. Martinelli},
{\em On the density of states of Schr\"odinger operators with a random potential}, J. Phys. A {\bf 15} (1982),  2139–-2156.

\bibitem{km1} {\sc W. Kirsch, F. Martinelli},
{\em Large deviations and Lifshitz singularity of the integrated density of states of random Hamiltonians},
Comm. Math. Phys. {\bf 89} (1983), 27–-40.

\bibitem{kme} {\sc W. Kirsch, B. Metzger}, {\em The integrated density of states for random Schrödinger operators}, in: F. Gesztesy, B. Simon (eds.):
Spectral theory and mathematical physics. Providence, RI: American Mathematical Society (2007), 649--696.

\bibitem{ks} {\sc W. Kirsch, B Simon},
{\em Lifshitz tails for periodic plus random potentials},
J. Statist. Phys. {\bf 42} (1986), 799–-808.

\bibitem{kw} {\sc W. Kirsch, S. Warzel}, {\em Anderson localization and Lifshits tails for random surface potentials}, J. Funct. Anal. {\bf 230} (2006),  222–-250.

\bibitem{klopp1} {\sc F.Klopp}, {\em  Internal Lifshits tails for random perturbations of periodic Schrödinger operators}, Duke Math. J. {\bf 98} (1999), 335--396.
\bibitem{klopp} {\sc F.Klopp}, {\em Lifshitz tails for alloy-type models in a constant magnetic field}, J. Phys. A {\bf 43} (2010), 474029, 9 pp.

\bibitem{klorai} {\sc F.Klopp, G.D.Raikov}, {\em Lifshitz tails in constant
magnetic fields}, Commun. Math. Phys., {\bf 267} (2006),  669-701.

\bibitem{kloppW} {\sc F.Klopp, T. Wolff} {\em: Lifshitz tails for 2-dimensional random Schrödinger operators}, J. Anal. Math. {\bf 88} (2002),  63--147.

\bibitem{l} {\sc L. Landau}, {\em Diamagnetismus der Metalle}, Z. Physik
  {\bf 64} (1930), 629--637.

\bibitem{ls}   {\sc H. Leinfelder, C. Simader},
Schr\"odinger operators with singular magnetic vector potentials.
Math. Z. 176 (1981),  1–-19.

\bibitem{mohrai}   {\sc A. Mohamed, G. D. Raikov},
{\em On the spectral theory of the Schrödinger operator with electromagnetic potential}, Pseudo-differential calculus and mathematical physics, 298–390,
Math. Top., {\bf 5}, Akademie Verlag, Berlin, 1994.




\bibitem{pf} {\sc  L. Pastur, A. Figotin}, {\em Spectra of Random and Almost-Periodic Operators}, Grundlehren der Mathematischen Wissenschaften {\bf 297}  Springer-Verlag, Berlin, 1992.

\bibitem {rs1} {\sc M. Reed, B. Simon}, {\it Methods of Modern
    Mathematical Physics I: Analysis of Operators},  Academic Press, 1978.

    \bibitem {rs4} {\sc M. Reed, B. Simon}, {\it Methods of Modern
    Mathematical Physics IV: Functional Analysis},  Academic Press, 1972.

 \bibitem{shu} {\sc M.A.Shubin},
{\em Pseudodifferential Operators and Spectral Theory}, Second
Edition, Berlin etc.: Springer-Verlag (2001).

  \bibitem {s1} {\sc B. Simon}, {\em The bound state of weakly coupled Schr\"odinger operators in one and two dimensions}, Ann. Physics {\bf 97} (1976), 279–-288.

 \bibitem {s} {\sc B. Simon}, {\em Trace Ideals and Their Applications}, Second edition. Mathematical Surveys and Monographs, {\bf 120} American Mathematical Society, Providence, RI, 2005.\\
\end{thebibliography}
\end{document}